\documentclass[9pt,onecolumn,superscriptaddress,floatfix,nofootinbib]{revtex4-2}


\usepackage{graphicx}%
\usepackage{multirow}%
\usepackage{lipsum}
\usepackage{amsmath}
\usepackage{amsfonts, amssymb, bbm, braket, accents}
\usepackage{graphicx}

\usepackage{booktabs}
\usepackage[dvipsnames]{xcolor}
\usepackage{braket}
\usepackage{quantikz}
\usepackage{soul}
\usepackage{float}
\makeatletter
\let\newfloat\newfloat@ltx
\makeatother
\usepackage{algorithm}
\usepackage{array}
\usepackage[makeroom]{cancel}
\usepackage{booktabs}
\usepackage{xcolor}
\usepackage{siunitx}
\usepackage{float}

\definecolor{cellmin}{rgb}{1,1,1}
\definecolor{cellmax}{rgb}{0.25,1,0.25}

\newcommand{\probP}{\text{I\kern-0.15em P}}
\newcommand{\probE}{\text{I\kern-0.15em E}}

\usepackage[breaklinks=true]{hyperref}
\hypersetup{
  colorlinks   = true,
  urlcolor     = blue,
  linkcolor    = blue,
  citecolor   = red
}
\usepackage{amsthm}

\usepackage[capitalise]{cleveref}
\crefformat{equation}{Eq.~(#2#1#3)}
\crefformat{section}{Sec.~#2#1#3}
\Crefformat{equation}{Equation~(#2#1#3)}
\crefformat{figure}{Fig.~#2#1#3}
\crefrangeformat{equation}{Eqs.~#3(#1)#4--#5(#2)#6}
\Crefformat{section}{Section~#2#1#3}

\newtheorem{theorem}{Theorem}
\newtheorem{corollary}[theorem]{Corollary}%
\newtheorem{definition}{Definition}%

\begin{document}
\title{Expressivity of IQP Feature-Map Quantum Kernels via Reproducing Kernel Hilbert Spaces}


\author{Valeria Repetto}
\email{valeria.repetto@iit.cnr.it}
\affiliation{National Research Council of Italy (CNR), Institute of Informatics and Telematics (IIT),via Giuseppe Moruzzi 1,Pisa, Italy}
\affiliation{Computational Biomedicine Unit, Department of Medical Sciences, University of Turin,via Santena 19,Torino,Italy}
\author{Stefano Melchionna}
\affiliation{Kipu Quantum, Greifswalderstrasse 212, 10405 Berlin, Germany}

\author{Romina D'Aurizio}
\affiliation{National Research Council of Italy (CNR), Institute of Informatics and Telematics (IIT),via Giuseppe Moruzzi 1,Pisa, Italy}

\date{\today}

\begin{abstract}
Quantum kernels are among the most promising methods in quantum machine learning, providing a natural way to embed data into high-dimensional Hilbert spaces via quantum circuits and capturing non-linear relationships in data. Their learning behaviour is governed by the function class they represent, which is formally characterized by the structure of the associated Reproducing Kernel Hilbert Space (RKHS).  However, a theoretical characterization of this function class, and of how it depends on the choice of feature map and its hyperparameters, is still lacking. In this work, we focus on the Instantaneous Quantum Polynomial (IQP) feature map and derive a Fourier representation of the associated RKHS generators as finite sums of trigonometric functions evaluated on polynomially transformed input variables, with a finite set of accessible frequencies determined by the circuit construction and scaled by the bandwidth hyperparameter. This representation identifies three bandwidth regimes with qualitatively distinct expressive properties. At small bandwidth the generators are localised and the kernel behaves comparably to a Gaussian kernel; at intermediate values they become multimodal, assigning high similarity to points that are distant in input space; at large values their periodic character dominates, and the kernel resolves periodic targets whose frequencies match those induced by the feature map. Numerical experiments on synthetic datasets with controlled periodic and multimodal structure, across data dimensions and feature-map configurations, reproduce the three regimes in both supervised classification and unsupervised anomaly detection, benchmarked against a classical Gaussian kernel.
These results make the relationship between expressivity and data structure explicit for IQP kernels: expressive power translates into a practical benefit only when the accessible frequencies match with the structure of the data, which provides a criterion for selecting the bandwidth when the periodic or multimodal structure of data is known.
\end{abstract}

\maketitle

\section{Introduction}\label{sec:Introduction}

Quantum Machine Learning (QML) is a rapidly growing field  that has attracted significant attention thanks to its potential to unlock new representational power beyond classical means \cite{Biamonte2016QuantumML}. In this landscape Quantum Kernels (QKs) have emerged as a promising method in QML \cite{Fuster_2024, Schuld2018QuantumML, JaegerKrems2023} since they offer an effective means of capturing  non-linear relationships in data by evaluating pairwise similarities within high-dimensional feature spaces naturally induced by quantum state encoding. Their versatility allows their application across a broad range of tasks, including classification \cite{Havlek2018SupervisedLW, SchnabelReview2025}, regression \cite{PaineQKMForDiffEQ2023}, clustering \cite{Repetto2025}, and anomaly or novelty detection \cite{Liang2019QuantumAD_OCSVM}. QKs have been tested so far in different domains including biology \cite{Vasques2023}, precision medicine \cite{Repetto2025}, solving differential equations \cite{PaineQKMForDiffEQ2023}, cosmology \cite{Peters2021MachineLO}, and finance \cite{miyabe2023IBM_HSBC, HSBC_2025, Grossi_QSVM_2022, Djehiche2021}.
There are empirical studies reporting both advantages and disadvantages of QKs with respect to classical kernel methods which depend on the dataset, feature map, and choice of hyperparameters \cite{Schuld_QA_2022, SchnabelReview2025}. This variability highlights the need to understand when the function class induced by a QK is useful for a particular learning problem and how the choice of feature-map hyperparameters should depend on the structural properties of the data. 
In this respect, \cite{Kbler2021TheIB} showed that the learning properties of QKs depend not only on the complexity and dimensionality of the underlying feature space, but also on how the functions represented by the kernel relate to the structure of the learning problem.
From a complementary perspective, \cite{Belis2026} suggested Fourier analysis as a natural framework for characterizing the frequency content of learning models.
Yet for one of the most widely studied families of QKs, namely those generated by the Instantaneous Quantum Polynomial (IQP) feature map \cite{Havlek2018SupervisedLW}, a theoretical characterization of the corresponding Reproducing Rernel Hilbert Space (RKHS) generators and of the role played by the feature-map hyperparameters is still lacking. This is particularly relevant for the bandwidth hyperparameter which controls the rescaling of the input data and has been linked to the model's ability to generalise \cite{Shaydulin2021ImportanceOK}, an effect later related to the kernel's eigenvalue spectrum and empirically confirmed on IQP kernels \cite{Canatar2022BandwidthEG}.
In this study, we address this gap by extending the framework of \cite{Schuld_data_encoding} to derive an explicit Fourier representation of the RKHS generators associated with IQP feature-map kernels. Our investigation shows that these functions are finite sums of trigonometric terms evaluated on polynomially transformed input variables. Additionally, we show how the bandwidth hyperparameter determines the characteristic frequencies of the kernel by identifying three distinct bandwidth regimes with qualitatively different expressive properties. Finally, we corroborate our theoretical results through numerical experiments on synthetic datasets with controlled structural characteristics (periodicity and multimodality) in both supervised classification and unsupervised anomaly detection tasks, using a classical Gaussian kernel as a baseline.
The paper is organised as follows. In Section \ref{sec:background}, we introduce the notation, definitions, and previous findings which constitute the foundation of our further investigation. Section \ref{sec:RHKS_fourier} presents the main theoretical results: we prove a Fourier sum representation of the RKHS generators and discuss the effect of the bandwidth parameter on the expressivity of kernel models. We then describe numerical experiments associated with the theoretical findings in Section \ref{sec:num_exp}. The last section is dedicated to the conclusions. 
\section{Background on kernel methods}\label{sec:background}
In this section, we introduce the notation, definitions, and related works  that will be used throughout the remainder of the paper. We begin by recalling the basic concepts of kernel methods, followed by a description of the connection between the RKHS and model expressivity. Finally, we introduce and discuss the properties of the IQP feature map based kernel, and of the classical Gaussian kernel, used as the classical counterpart for the numerical experiments (Section \ref{sec:num_exp}).

\subsection{Kernel methods}
Kernel-based learning methods leverage the idea of a similarity measure between data points \cite{Scholkopf2001}.
\begin{definition} [Kernel]
    A kernel is a positive semi-definite bi-variate function 
    $$
    k: \mathcal{X} \times\mathcal{X} \to \mathbb{R} 
    $$    
    defined on an input domain $\mathcal{X}$.
\end{definition}
Every positive semidefinite kernel function can be interpreted as the inner product of data points mapped into a (possibly infinite-dimensional) Hilbert feature space $\mathcal{H}$ by a feature map $\phi:\mathcal{X}\rightarrow\mathcal{H}$, namely
\begin{equation}\label{eq_inner_pro}
k(x_1, x_2) = \langle \phi(x_1), \phi(x_2) \rangle_\mathcal{H} .
\end{equation}
The practical consequence of \eqref{eq_inner_pro}, known as the \emph{kernel trick}, is that rather than constructing the feature map explicitly, kernel methods operate directly on kernel evaluations, thereby avoiding computations in the feature space while implicitly exploiting its geometry \cite{Scholkopf2001}. This enables linear algorithms in $\mathcal{H}$ to represent highly non-linear relationships in the original input space. Many kernel-based learning algorithms can therefore be expressed exclusively in terms of kernel evaluations. Given a training set $\{x_i\}_{i=1}^n$, a kernel model takes the form
\begin{equation}\label{eq_linear_kernel_model}
    m(x) = z \left( \sum_i \alpha_i k(x_i, x) + \beta \right)
\end{equation}
where $z: \mathbb{R} \to \mathcal{Y}$ is some function (e.g., the sign function in a binary classification problem) and $\alpha_i$ and $\beta$ are real valued coefficients learned via model fitting. Note that the formulation described by (\eqref{eq_linear_kernel_model}) includes a variety of models such as Support Vector Machines (SVMs) \cite{SVM} and One-Class SVMs (OCSVM) \cite{Schoelkopf_OCSVM}. \\
Here and in the rest of the paper, we assume our input feature domain to be $\mathcal{X} = [0,1]^D \subset \mathbb{R}^D$ for $D \in \mathbb{N}$, i.e. bounded in the unit hypercube, without being overly restrictive, as one can always rescale input data. We assume the output space $\mathcal{Y}$ to be a subset of $\mathbb{R}^M$. Examples are $\mathbb{R}$ for regression problems or anomaly detection (where the output is interpreted as an anomaly score), the multi-binary set $\{0,1\}^M$ for classification and clustering problems with $M$ indicating the number of classes or clusters. \\
\subsection{Model expressivity and RKHS}
The notion of \emph{expressivity} is widely used in ML, although no universally accepted mathematical definition exists. For QK methods it is often associated with the dimension of the Hilbert space into which data are embedded. However, a large embedding space does not by itself imply better learning since only the image of the feature map is accessible, and, more importantly, the functions favoured by the kernel have to be aligned with the learning task \cite{Kbler2021TheIB}. 
Throughout this work, we adopt a functional notion of expressivity: for a fixed kernel, we call expressivity the class of functions that can be represented by the corresponding kernel model \eqref{eq_linear_kernel_model} when its trainable coefficients and hyperparameters vary. Under this definition, a more expressive kernel can represent a broader class of target functions. 
Function classes are only partially ordered by inclusion, so expressivity cannot be quantified by a single scalar measure.
For the kernel models considered in this work, this notion of expressivity is naturally characterised through the RKHS associated with the kernel (see \cite{petruccione_book}, Definition 6.4).
\begin{definition} [RKHS]
    The RKHS of a kernel $k$ over $\mathcal{X}$ is the Hilbert space created by completing the span of functions $f : \mathcal{X} \to \mathbb{R}$, $f(\cdot) = k(x,\cdot)$ for $x \in \mathcal{X}$ (i.e., including the limits of all Cauchy sequences). The inner product is defined by 
    $$
    \langle f, g \rangle = \sum_{i,j} a_i b_j k(x_i, x_j) 
    $$
    for two functions $f(\cdot) = \sum_i a_i k(x_i,\cdot))$ and $g(\cdot) = \sum_j b_j k(x_j,\cdot))$ and extended by linearity. The functions $k(x,\cdot)$ are called generators.  
\end{definition}
The RKHS describes the space of functions which are achievable as (limits of) linear combinations of kernel generators and it determines which functions a kernel model can represent. Properties of the generators $k(x, \cdot)$ thus affect the properties of kernel models and therefore their expressivity.  

\subsection{Gaussian Kernels}
Several different classical kernels are commonly used in ML. The Gaussian kernel, also known as Radial Basis Functions (RBFs) kernel, is one of the most popular and widely used classical kernels. Its popularity stems from its strong performance across a variety of  tasks, as well as its ease of implementation and hyperparameter tuning. Additionally, RBF kernels account for some desirable properties discussed below. For these reasons we employed Gaussian kernel models as the classical case for our experiments in the following sections. \\

The generators of the RKHS for Gaussian kernels are Gaussian functions:
$$
    \exp\left(-\gamma \|x_1 - x_2\|_2^2\right), \qquad \gamma \in \mathbb{R}_+.
$$
As a consequence Gaussian kernel models are translation and rotation invariant, their kernel matrix is a similarity matrix, where every entry is bounded and the similarity of any point with itself is maximal, i.e., $k(x_1,x_1)\geq k(x_1,x_2)$ for all $x_1,x_2 \in \mathcal{X}$. The Gaussian RBF kernel represents two data points $(x_1,x_2)$ as similar ($k(x_1,x_2)$ is high) only if they are close in the original feature space $\mathcal{X}$. The strength of the interaction with respect to the euclidean distance between points in the space $\mathcal{X}$ is mediated by the hyperparameter $\gamma$.
Generators of RKHS for Gaussian kernels are displayed in Figure \ref{fig:RKHS_gaussian}.
\begin{figure}[!h]
    \centering
    \includegraphics[width=0.9 \textwidth]{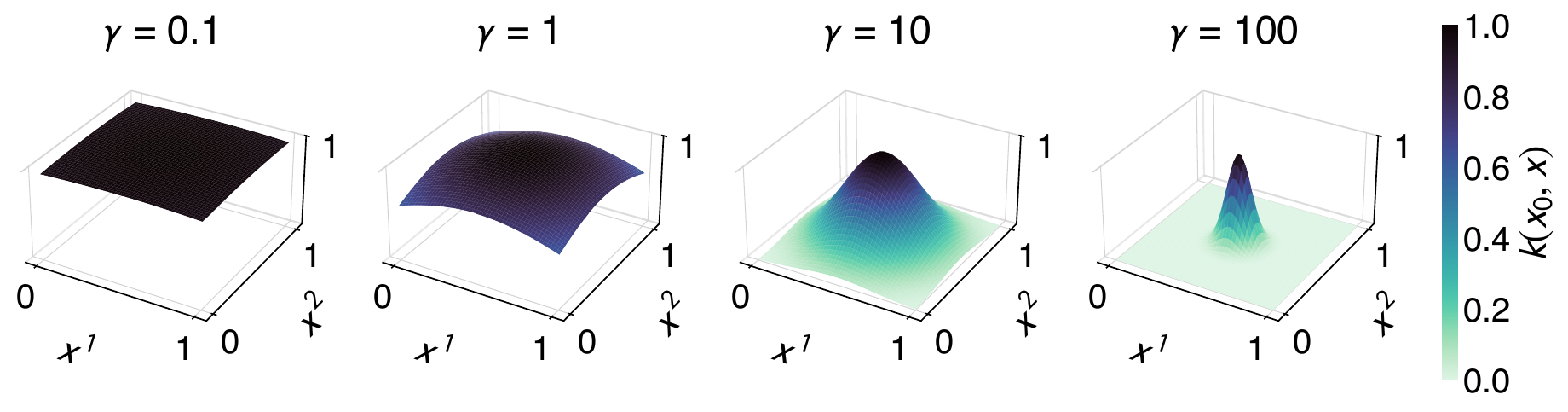}
    \caption{Gaussian RKHS generators for different values of $\gamma$. Each plot displays the values of the kernel $k(x, c)$ where $c$ is fixed at the centre of the square and $x = (x^1, x^2)$ where $x^1$ and $x^2$ are the Cartesian coordinates on the horizontal plane.}
    \label{fig:RKHS_gaussian}
\end{figure}

\subsection{Quantum Kernels and IQP feature map} \label{sec:quantum_kernels}
In QK the input, data $x$ are encoded, via an encoding function $\phi(x)$, in a quantum state. Formally, this encoding process is realised through the application of a unitary transformation $U_\phi (x)$ acting on $D$ qubits in the initial ground state $\ket{0} ^{\otimes D} $ 
$$
\phi : x \to \ket{\phi(x)} = U_{\phi (x)} \ket{0} ^{\otimes D} 
$$
Given $x_1, x_2 \in \mathcal{X}$ the kernel is computed as 
$$
k(x_1, x_2) = \bra{\phi(x_2)}\ket{ \phi(x_1)}^2 = \bra{0}^{\otimes D} U_{\phi (x_2)}^{\dagger} U_{\phi (x_1)} \ket{0} ^{\otimes D}
$$

A wide range of choices is available for the feature map $\phi$, and selecting an appropriate quantum embedding is a non-trivial task as it determines both the structural characteristics of the kernel and the practical implementability of the resulting model \cite{Cerezo2022ChallengesAO}.
In this work, we adopt the Instantaneous Quantum Polynomial (IQP) feature map introduced in \cite{Havlek2018SupervisedLW}. Besides being one of the most studied quantum feature maps, it yields fidelity kernels that are believed to be hard to simulate classically \cite{Incudini2024TowardUQ}.
The IQP feature map embeds an input vector $x = (x^1, \dots, x^D)$ into a quantum state by first applying a layer of Hadamard gates, followed by data-dependent single-qubit rotations whose angles are determined by the individual input components $x^i$. These are combined with entangling operations and additional single-qubit rotations whose angles depend on pairs of input components $x^i, x^j$. Different choices of the entangling pattern lead to different feature maps. The two most commonly adopted architectures are \emph{linear entanglement}, in which each qubit interacts only with its nearest neighbours, and \emph{full entanglement}, in which every pair of qubits is entangled. 
The corresponding unitary $U_{\phi(x)}$  is given by the following equation:
$$
U_{\phi(x)} = \prod_{[i,j] \in S} R_{Z_i, Z_j}(x_i, x_j) \bigotimes_{k=1}^D R_{Z_k}(x_k)
$$
where  $R_{Z_i, Z_j}$  and $R_{Z_k}$  represent two and one qubits rotations and $S$ is the set of gates to be entangled. 
IQP style feature maps can be applied $r$ times, called repetitions, allowing for a  more complex encoding. An high level representation of the IQP circuit module is shown in Figure  \ref{fig:IQP} (a) for linear entanglement and Figure  \ref{fig:IQP} (b) for full entanglement.
Encoding a point $x \in R^D$ into a quantum state $\ket{\phi(x)}$ projects the input data into a higher dimensional Hilbert space \cite{Schuld2021QuantumML}. Often the high dimensionality leads to kernel matrices with entries concentrated around the diagonal \cite{Thanasilp2024ExponentialCA, Suzuki2024} hindering subsequent model training. It has been shown \cite{Shaydulin2021ImportanceOK, Canatar2022BandwidthEG} that the introduction of a tunable hyperparameter $b \in\mathbb{R}_+ $, called bandwidth, controlling how input data are rescaled before being encoded into the quantum feature map, can enhance QK models generalisation capabilities.
Points in the original domain $[0,1]^D$ are mapped to $[0, b\pi]^D$ via $x \;\mapsto\; b\pi x$. For a single feature $x^i \in [0,1]$, the corresponding single-qubit encoding applies a rotation of angle $2\pi b x^i$. Hence, the encoding is periodic with period $1/b$ and frequency $b$. The bandwidth tunes the sensitivity of the kernel to variations in the data. For small $b$ (low frequency), the encoding varies slowly with the input. Nearby points in the input space are mapped to very similar quantum states, leading to smoother kernels that capture coarse-scale structures. For large $b$ (high frequency), the encoding varies rapidly with the input. Even small differences in the data may result in very different quantum states, making the kernel sensitive to fine details. At the same time, in this regime, periodic effects also become visible due to the periodicity of the qubit rotations.

\begin{figure*}
\includegraphics[width=1 \textwidth]{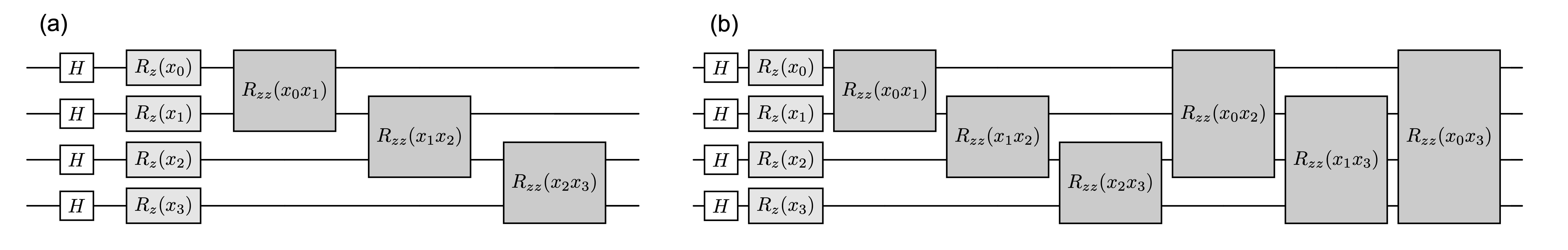}
\caption{Graphic high-level circuit representation of 4 qubit IQP Feature Map with (a) linear entanglement and (b) full entanglement.}
  \label{fig:IQP}
\end{figure*}

\section{Partial Fourier sum representation of IQP-feature maps based kernels} \label{sec:RHKS_fourier}

In this section, we prove that the RKHS generator of IQP-feature map based kernels is a combination of trigonometric functions of the input data $x_i$ and of the products $x_ix_j$. This result is formalized in the Theorem \ref{theorem:RHKS_fourier} and Corollary \ref{corollary:IQP_RHKS_fourier} which follow and generalise the arguments outlined in \cite{Schuld_data_encoding} and \cite{petruccione_book}. \\

\begin{theorem} \label{theorem:RHKS_fourier}
    Let $x \in \mathbb{R}^D$ and $\phi(x)$ be a feature map acting on $D$ qubits which alternates $L \in \mathbb{N}$ layers, each composed by data independent unitaries $W^{l}$ acting on multiple qubits and data embedding layers $S^{l}(x) = \bigotimes _{d=1}^DS^{l}_d(x)$ acting independently on each single qubit:
    \begin{equation} \label{eq:IQP_layers_decomposition}
    \phi(x) = W^{L+1}S^L(x)W^{L}....S^1(x)W^{1}.
    \end{equation}
    Assume 
    \begin{equation} \label{eq:Sl}
        S^{l}(x) = e^{-ic^l_1(x)H}\otimes ... \otimes e^{-ic^l_D(x)H} 
    \end{equation}
    where $H$ is a fixed single-qubit operator with eigenvalues  $\lambda_1$  and $\lambda_2$, and the functions $c_d^l : \mathbb{R}^D \to \mathbb{R}$ are arbitrary real-valued functions encodings. Let 
    $k(x,x^\prime) = \bra{0}^{\otimes D} \phi (x^{\prime}) \phi (x) \ket{0} ^{\otimes D}.$ \\

    Then the kernel $k$ can be written as     
    \begin{equation} \label{eq:fourier_rep_kernel}
        k(x,x^\prime) = \sum_{\mathbf{j}}\sum_{\mathbf{k}} \alpha^{*}_{\mathbf{k}}\alpha_{\mathbf{j}}e^{-i \left( \mathbf{c}_{\mathbf{j}}(x)\mathbf{\lambda_{\mathbf{j}}} - \mathbf{c}_{\mathbf{k}}(x^{\prime})\mathbf{\lambda_{\mathbf{k}}} \right) }
    \end{equation}
    where $\mathbf{j}, \mathbf{k} \in [1, \dots , D] ^L$ are multi-indexes, $\mathbf{c}_{\mathbf{j}} = (c^l_j)_{l \in \{1, ..., L \} , j \in \mathbf{j}}$, $\mathbf{\lambda_{\mathbf{j}}} $ is a vector whose entries are eigenvalues of $H$, and $\alpha_{\mathbf{j}}$ are coefficients independent of the data $x$.     
\end{theorem}
\begin{proof}
    The proof closely follows \cite{Schuld_data_encoding} and \cite{petruccione_book}. For this reason, we begin by recalling the results established therein, where it was additionally assumed that $S^{l}_d(x) = e^{-ix_d H}$ is independent on the layer $l$ and acts only on the $d$-th coordinate of $x$ by means of a single qubit unitary $H$ as illustrated in Figure \ref{circ:asbtract_z_map}. An example of such feature map is the Z-feature map where $H$ is the Pauli-Z unitary $H = \sigma_z$. 

    \begin{figure}[h]
    \centering
    \begin{quantikz}
    &&\gate[2]{W^{1}}&\gate{e^{-ix_1 H}}& \gate[2]{W^{2}} &\gate{e^{-ix_1 H}} & \ldots\ &  \gate[2]{W^{L}} \\
    &&& \gate{e^{-ix_2 H}} &&  \gate{e^{-ix_2 H}} & \ldots\ &
    \end{quantikz}
    \caption{Circuit representation of feature maps considered in \cite{Schuld_data_encoding} and \cite{petruccione_book}. An example of such map is the Z-feature map.}
    \label{circ:asbtract_z_map}
    \end{figure}
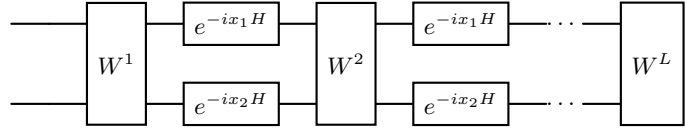
    With these additional assumptions, from Theorem 6.1 of  \cite{Schuld_data_encoding} and Corollary 6.1 of \cite{petruccione_book} the kernel generated by such a feature map can be expressed as a partial Fourier sum:
    $$
    k(x,x^\prime) = \sum_{n,n^\prime} e^{-inx}e^{in^\prime x^\prime} c_{n, n^\prime}  
    $$
    where $n,n^\prime$ are vectors of eigenvalues of $H$. Of note, in the case of the Z-feature map, $H = \sigma_z$ has eigenvalues $-1$ and $+1$, therefore the possible frequencies for the RKHS generator $K(\cdot, x^\prime)$ are in the set of differences of pairs of eigenvalues of $H$: $\Omega = \{-2,0,2\}$. \\
    In our setting $S^{l}(x)$ is more general and defined by \eqref{eq:Sl}. An illustration of the corresponding circuit is given in Figure \ref{circ:general_feat_map} for $D= 2$.

    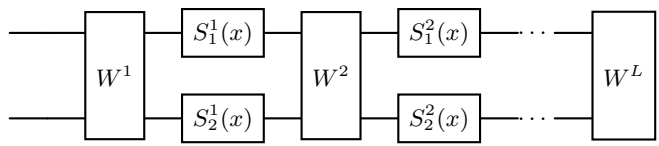
\begin{figure}[h]
    \centering
    \begin{quantikz}
    &&\gate[2]{W^{1}}&\gate{S^{1}_1(x)}& \gate[2]{W^{2}} &\gate{S^{2}_1(x)} & \ldots\ &  \gate[2]{W^{L}} \\
    &&& \gate{S^{1}_2(x)} &&  \gate{S^{2}_2(x)} & \ldots\ &
    \end{quantikz} 
    \caption{Circuit representation of an abstraction of the IQP-feature map for two qubits.}
    \label{circ:general_feat_map}
    \end{figure}

    We now want to compute the $j$-th component of $\phi(x) \ket{0}^{\otimes D} $. 
    By assuming, without loss of generality, that $H$ is diagonal and has eigenvalues $\{ \lambda _i \}_{i=1,2}$ and noting therefore that $\left[ S^{l}(x) \right]_{j,k} = \delta_{j,k}  e^{-ic^l_j(x)\lambda_{\tilde{j}}} $ for some eigenvalue $\lambda_{\tilde{j}}$ of $H$, we have 
    \begin{equation*}
        \left. \phi(x) \ket{0}^{\otimes D} \right] _j =    \sum_{j_1,...j_L=1}^DW^{L+1}_{j,j_L}W^{L}_{j_L,j_{L-1}} \dots W^{1}_{j_1,1} e^{-ic^L_{j_L}(x)\lambda_{\tilde{j_{L}}}} \dots e^{-ic^1_{j_1}(x)\lambda_{\tilde{j_1}}}.
    \end{equation*}
    By introducing the multi-index notation $\mathbf{j} \in [1, \dots , D] ^L$, we get
    \begin{equation*}
        \phi(x) \ket{0}^{\otimes D}  =    \sum_{\mathbf{j}}\alpha_{\mathbf{j}} e^{-i \mathbf{c}_{\mathbf{j}}(x)\mathbf{\lambda_{\mathbf{j}}}}
    \end{equation*}
    where $\alpha_{\mathbf{j}}$ are coefficients independent of the data $x$, $\mathbf{c}_{\mathbf{j}} = (c^l_j)_{l \in \{1, ..., L \} , j \in \mathbf{j}}$ and $\mathbf{\lambda_{\mathbf{j}}} $ is a vector whose entries are eigenvalues of $H$. 
    As a consequence the kernel $k$ generated by the feature map $\phi $ can be expressed as 
    
    \begin{equation*}
        k(x, x^{\prime}) = \bra{0}^{\otimes D} \phi (x^{\prime}) \phi (x) \ket{0} ^{\otimes D}
         = \sum_{\mathbf{j}}\sum_{\mathbf{k}} \alpha^{*}_{\mathbf{k}}\alpha_{\mathbf{j}}e^{-i \left( \mathbf{c}_{\mathbf{j}}(x)\mathbf{\lambda_{\mathbf{j}}} - \mathbf{c}_{\mathbf{k}}(x^{\prime})\mathbf{\lambda_{\mathbf{k}}} \right) }.
    \end{equation*}
\end{proof}

The following result can be derived directly from Theorem \ref{theorem:RHKS_fourier}. 

\begin{corollary} \label{corollary:IQP_RHKS_fourier}
    The RKHS generator associated with the IQP-feature map can be expressed as the composition of a non-linear function $c: \mathbb{R}^D \to (\mathbb{R}^D)^L$ where every entry of this function is a polynomial of degree $2$ (or lower) - which can also be regarded as a change of coordinates - and a Fourier like sum over the domain $(\mathbb{R}^D)^L$ where the accessible frequencies are integer multiples of $2$. 
\end{corollary}
\begin{proof}
    This result directly derives from Theorem \ref{theorem:RHKS_fourier} by noting that the IQP-feature map can be written as in \eqref{eq:IQP_layers_decomposition} and \eqref{eq:Sl} where $H = \sigma_z$ and 
    $c^1(x)= x$, i.e., $c^1_d(x)= x_d$, and either $c^l_d(x)= (\pi - x_d)(\pi - x_{d^\prime})$ for some $d^\prime$ or $c^l_d(x)=0$ for $l>1$. Here we interpret $e^{-i 0 H}$ as the identity matrix. Since $H = \sigma_z$ has eigenvalues $-1$ and $+1$, each encoding gate contributes elementary frequency differences in $\{-2,0,2\}$. The total phase accumulated by the feature map is obtained by combining the contributions from all encoding gates, so the accessible frequencies are finite sums of elements of $\{-2,0,2\}$, that is, integer multiples of $2$. A graphical representation of this feature map for $2$ qubits is given in Figure \ref{circ:IQP_feat_map}.
    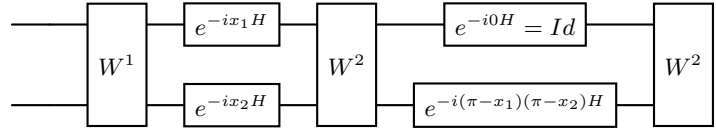
\begin{figure}[h]
    \centering
    \begin{quantikz}
    &&\gate[2]{W^{1}}&\gate{e^{-ix_1 H}}& \gate[2]{W^{2}} &\gate{e^{-i 0 H} = Id} &   \gate[2]{W^{2}} \\
    &&& \gate{e^{-ix_2 H}} &&  \gate{e^{-i(\pi -x_1)(\pi - x_2) H}} & 
    \end{quantikz}
    \caption{Representation of the IQP-feature map circuit for 2 qubits and and 1 repetition.}
    \label{circ:IQP_feat_map}
    \end{figure}
\end{proof}

As regards the dependence of the RKHS on the bandwidth hyperparameter, for $D=2$, the RKHS generators can be represented graphically as in Figure \ref{fig:RKHS_quantum}, where the bandwidth varies between 0.125 and 4.00. 
The trigonometric-like structure of those kernels derives from Corollary \ref{corollary:IQP_RHKS_fourier} and recalling that the bandwidth $b$ acts as a scaling factor on the input data. The RKHS generators are finite sums of periodic functions composed with non-linear polynomial coordinate transformations. Therefore, they are generally not periodic in the original input coordinates, but inherit oscillatory structures whose characteristic spatial frequencies are controlled by $b$.
Building on this analytical characterisation, we identified three qualitatively different regimes as $b$ varies. For large values of $b$, the periodic nature of the components of the RKHS generator is clearly evident in Figure \ref{fig:RKHS_quantum}. 
Intuitively, this suggests that QK models better capture periodic patterns in data whose intrinsic periodicity is commensurate with the period implied by $b$. While, for small values of $b$, the kernel RKHS can capture only local interactions between data, similarly to the RBF kernel, since the input data $x$ are confined to a bounded domain. In this regime, QKs closely resemble classical RBF kernels consistently with the observations in \cite{Alban_2025}. Intermediate values of $b$ give rise to a distinct behaviour, which we describe as multimodal: the full underlying periodic structure is not yet visible, yet QKs can represent points as similar when they are distant in the original feature space thus capturing longer-range interactions. \\ 
We now examine the other hyperparameters: the number of repetitions of the embedding circuit and the entanglement strategy. 
The number of repetitions influences the accessible frequencies: higher $r$ implies multiple contributions for each term $c^l_d(x)$ in Equation \eqref{eq:fourier_rep_kernel}. The individual coefficients of each term $c^l_d(x)$ (given by differences of eigenvalues of $\sigma_z$) sum up yielding potentially integer multiples of the base frequency $2b$ (see also Figure \ref{fig:RKHS_quantum} and Figure \ref{fig:RKHS_Fourier}). 
The effect of the entanglement strategy is less immediate as it affects only the $W^l$  terms in Theorem \ref{theorem:RHKS_fourier}. As a direct consequence, changing the entanglement strategy does not modify the set of accessible frequencies in the Fourier-like representation of the RKHS generators, but only the corresponding expansion coefficients. It therefore regulates the relative weight of the different terms and frequencies, rather than introducing new frequencies. In this sense, different entanglement patterns modify the relative importance of the components already available in the RKHS representation.
\\
\begin{figure}[!h]
    \centering
    \includegraphics[width=0.9 \textwidth]{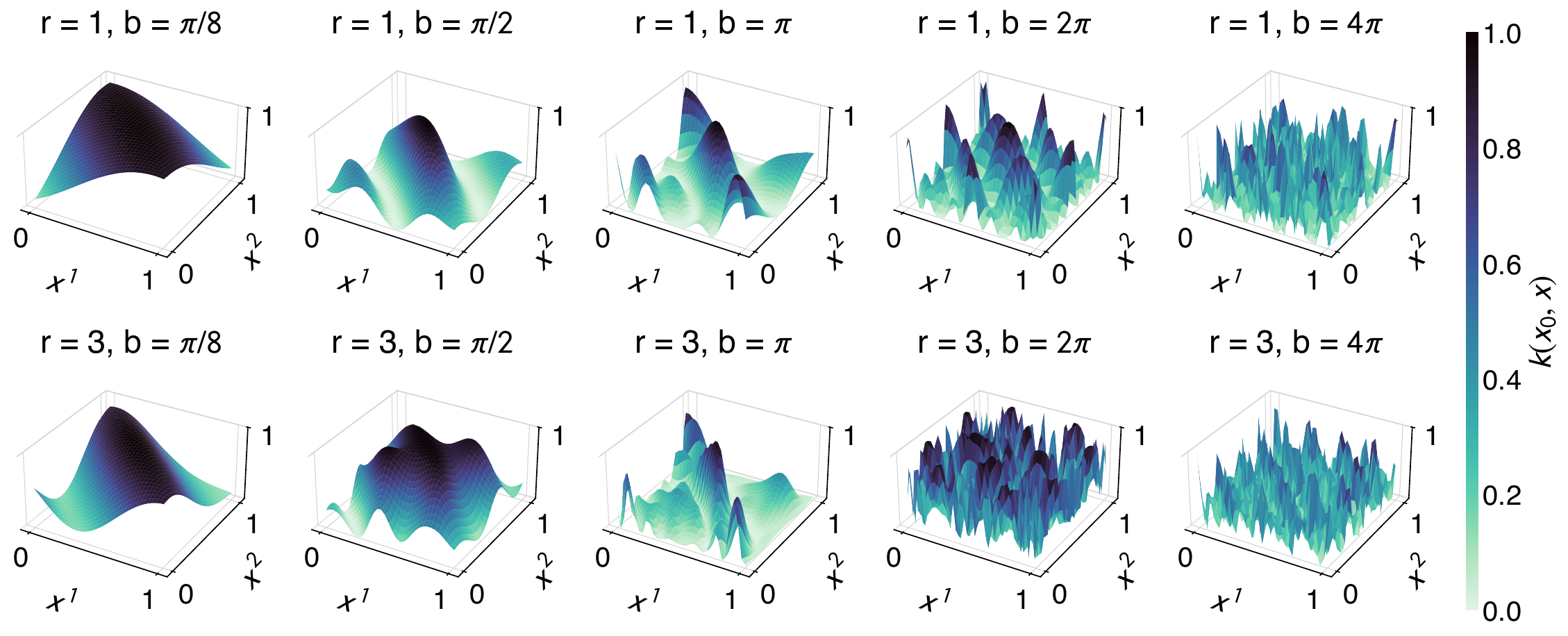}
    \caption{Quantum RKHS generators are represented for different values of bandwidth and number of repetitions. Each plot displays the values of the kernel $k(x, c)$ where $c$ is fixed at the centre of the square and $x = (x^1, x^2)$ where $x^1$ and $x^2$ are Cartesian coordinates in the horizontal plane. The periodic-like nature of the kernel becomes visible for larger values of $b$. It can also be observed that for intermediate values of $b$ the kernel is multi-modal which suggests that it can be used to capture multi-modality or long range interactions in the data. Increasing the number of repetitions increases the complexity of the RKHS generator.}
    \label{fig:RKHS_quantum}
\end{figure}

An alternative approach to analyse the pseudo-periodic structure described above is by looking at the Fourier transform of the RKHS generator. Given $x\in [0,1]^D$, the RKHS generator can be interpreted as an integrable function $k(x, \cdot) : [0,1]^D\to \mathbb{R}$. It can be therefore represented as a Fourier series with coefficients $c_{u}$ for integer frequencies $u \in \mathbb{Z}^D$. If the magnitude of $c_{u}$ is large it means that the RKHS generator has strong oscillations in the frequencies represented by $u$. Figure \ref{fig:RKHS_Fourier} represents the Fourier transform of the $2$D RKHS. Each pixel $(u^1,u^2)$ in the image represents the magnitude of the Fourier coefficient $c_{u^1,u^2}$ corresponding to the frequencies $(u^1,u^2)$. Note that, to reduce the large magnitude of the point $(0,0)$ in the Fourier space, the RKHS has been normalised to have mean value $0$ before applying the Fourier transform. It can be observed that the coefficients are large when $u^1$ and $u^2$ are proportional to the bandwidth $b$. This confirms that the RKHS generator has strong periodic components at frequencies proportional to the bandwidth $b$. Intuitively, this suggests that the QK models might be able to better capture periodic patterns in the data if the period of the data is aligned with the period implied by $b$.
\begin{figure}[!h]
    \centering
    \includegraphics[width=0.9 \textwidth]{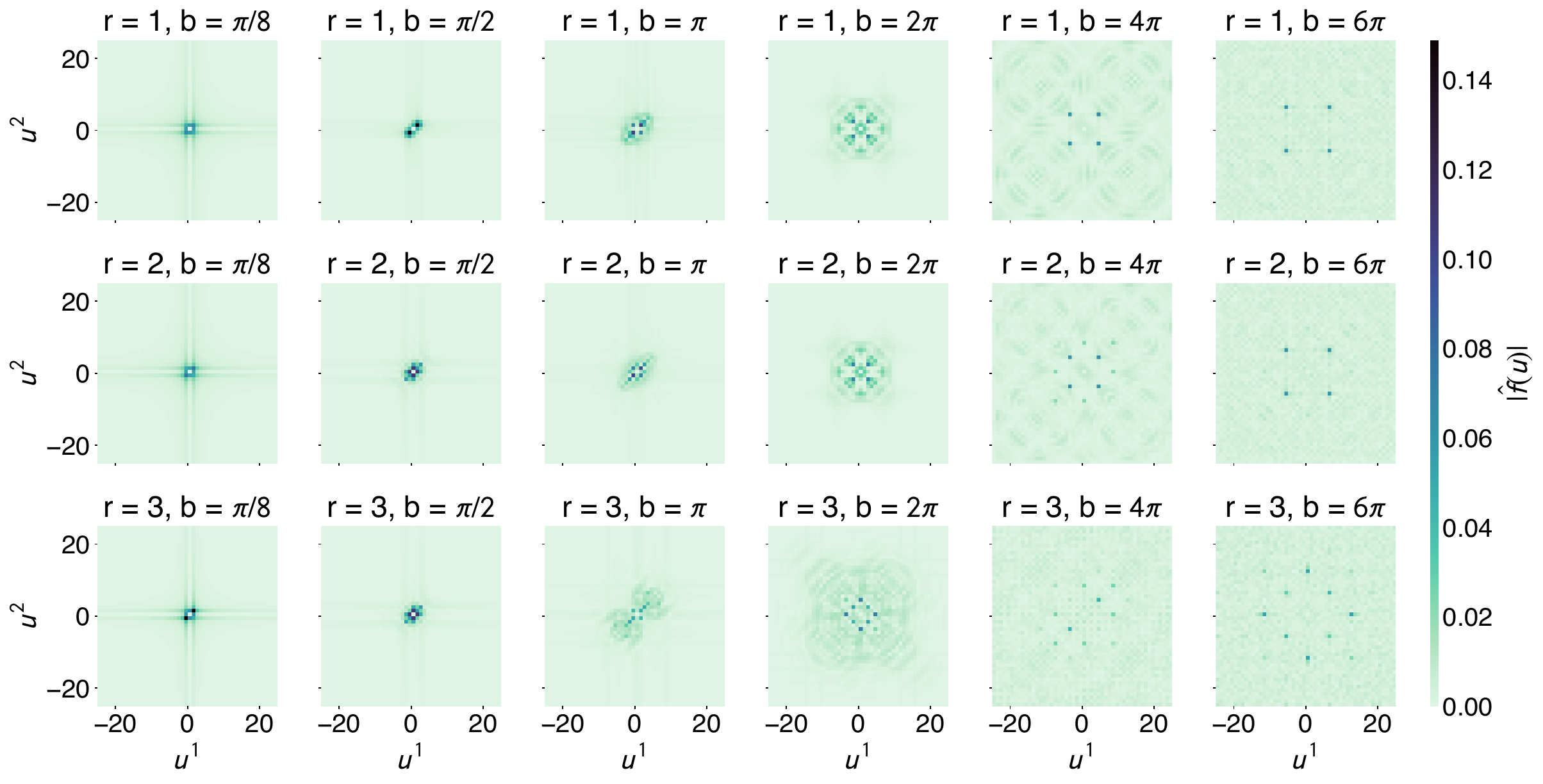}
    \caption{Heatmap representations of the Fourier transform of the $2$D RKHS (see Figure \ref{fig:RKHS_quantum}) for different values of the bandwidth and reps parameter with linear entanglement.  }
    \label{fig:RKHS_Fourier}
\end{figure}

 \section{Numerical experiments} \label{sec:num_exp}
In this section, we report two sets of numerical experiments designed to probe the expressivity of IQP-feature map QK on distinct bandwidth regimes identified in Section \ref{sec:RHKS_fourier}. The first experiment considers data with an underlying periodic structure (Section \ref{sec:periodic}) and investigates the high bandwidth regime. The second considers multimodal data (Section \ref{sec:multimodal}) and investigates the corresponding intermediate bandwidth regime. In both cases, we compare QK models against a classical Gaussian kernel baseline on supervised classification and unsupervised anomaly detection tasks using SVMs \cite{SVM} and One-Class SVMs (OCSVM) \cite{Schoelkopf_OCSVM, Liang2019QuantumAD_OCSVM} respectively. 
\subsection{Periodic data}  \label{sec:periodic}
To assess the ability of QK methods to learn data with a periodic structure, we generated the dataset by uniformly sampling random data in the unit hypercube $[0,1]^D$ for different values of $D \in \{2,4,6\}$. Data were labelled into two classes based on the sign of an oscillating function defined as follows: for all $x = (x^1, ..., x^D) \in [0,1]^D$ the label $y(x) \in \{0,1\}$ is computed as 
\begin{equation} \label{eq:classification_label}
   y(x) = \frac{1}{2} \left[ sign \left( \sum_{j=1}^D \sin \left(2 \pi\omega  x^j\right)\right) +1 \right]
\end{equation}
where $\omega \in \mathbb{N}$ represents the number of oscillations of the function in the domain, i.e., the frequency. In this experiment, we considered $\omega \in \{1,2,4,6,8\}$. The two classes used for the classification task were generated by Equation \eqref{eq:classification_label}, yielding an approximately balanced dataset. For the anomaly detection task, the majority class was sampled from the data labelled 0 and a number of points equal to $20\%$  of the majority class  was uniformly generated in the unit hypercube $[0,1]^D$ and added as the minority class. The periodic nature of the dataset can be analytically derived and is shown in Figure \ref{fig:datasets}(a,b).
\begin{figure}
    \centering
    \includegraphics[width=0.9 \textwidth]{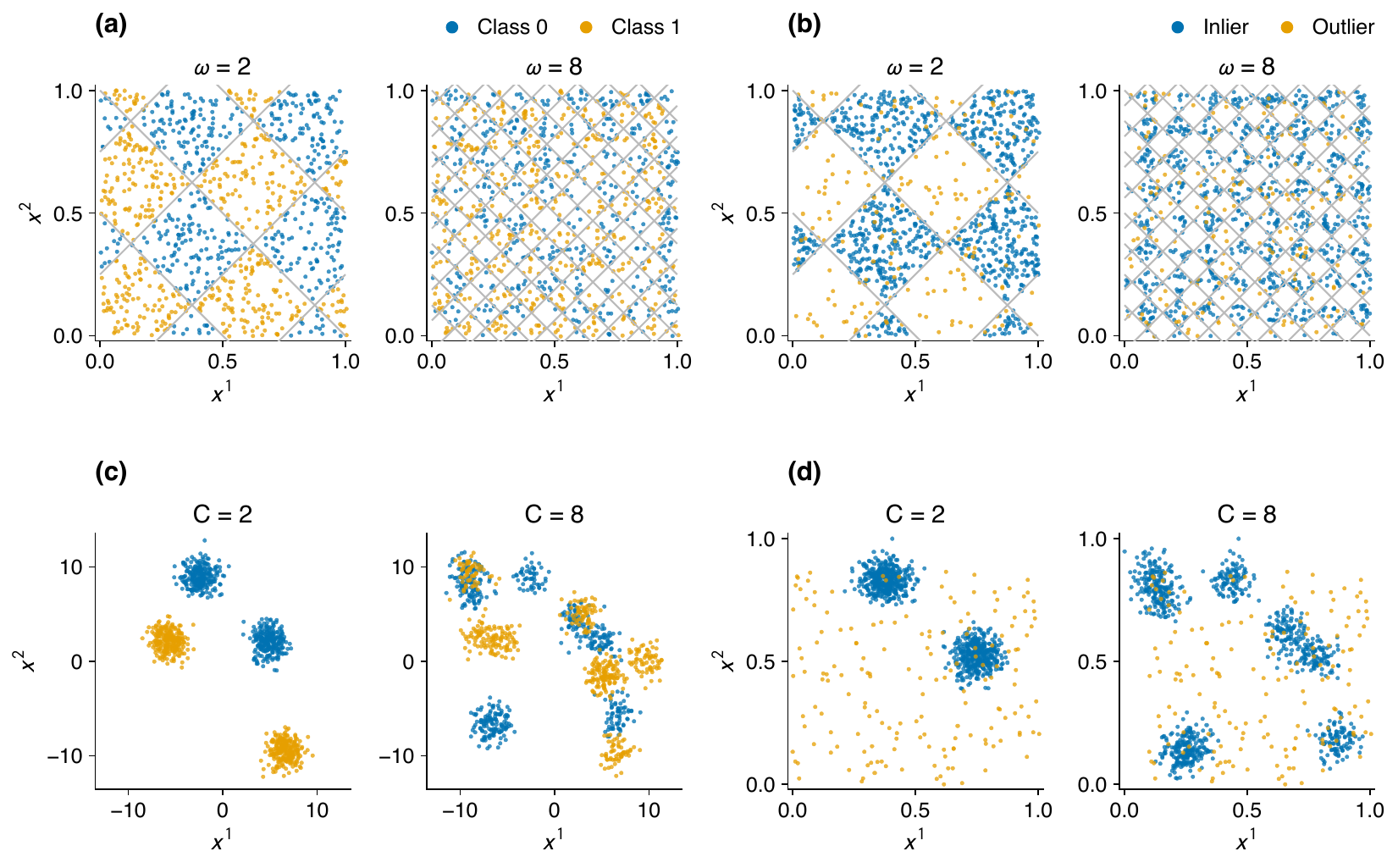}
    \caption{$2$D generated datasets. (a) Periodic datasets used for the classification task for $2$ and $8$ oscillations $\omega$. The grey lines indicate the class separation boundary. (b) Periodic datasets used for the anomaly detection task for $2$ and $8$ oscillations. The grey lines delineate the support of the inlier distribution. (c) Multimodal datasets used for the classification task with $2$ and $8$ centres. (d) Multimodal datasets used for the anomaly detection task with $2$ and $8$ centres. Each dataset contains 2000 points.}
    \label{fig:datasets}
\end{figure}

To benchmark different QK configurations, for each dataset and task, we computed QKs for different values of the bandwidth $b \in \{1,2,4,6,8\}$, of the number of repetitions of the data encoding feature map  $r \in \{1,2,3\}$, and for two different entanglement strategies, namely \textit{linear} and \textit{full}. In addition, we computed a classical Gaussian kernel using the scikit-learn implementation with default parameter $\gamma= 1/ (D \times  Var(X))$ \cite{scikit-learn} and additional values $\gamma \in \{0.01,0.1,1,10,100\}$. For each resulting kernel, an SVM was trained and evaluated on the classification task, and an OCSVM on the anomaly detection task. The F1-score was selected as main evaluation metric for both task as it is well suited for both balanced and imbalanced datasets. 

Figure \ref{fig:combined_periodic_classification}(a,b) reports the test-set F1-score over the full $(\omega,b)$ grid. For the best performing configuration (with 2 qubits, 3 repetitions, linear entanglement) the highest F1-score is concentrated on the diagonal $b=\omega$, i.e. when the bandwidth is aligned with the periodicity of the underlying data, and varies between 0.75 and 0.84, while it falls in 0.46-0.72 when $b \neq \omega$. The same trend is evident after averaging over all QK configurations (panel b), although with a general reduced magnitude (0.49-0.80), suggesting that it's not an artefact of a single favourable hyperparameters choice. 
When we consider the matched $b=\omega$ (Figure \ref{fig:combined_periodic_classification}(c)) adding the classical Gaussian kernel for comparison across different $D \in \{2,4,6\}$, we notice that, beyond the trivial case $\omega=1$, the average QK's F1-score is approximately stable in $\omega$ for all $D$: it ranges from $0.657$ to $0.691$ for $D=2$, from $0.548$ to $0.583$ for $D=4$, and from $0.536$ to $0.553$ for $D=6$. The overall F1-score decreases with $D$, reflecting the increased difficulty of modelling higher-dimensional periodic structures. The classical Gaussian kernel behaves qualitatively differently: at low frequency ($\omega \in \{1,2\}$) it better captures the structure of the data, but its F1-score decreases monotonically with $\omega$, crossing the QK mean at $\omega \geq 4$ for $D \in \{2,4\}$.
Also for the anomaly detection task, the QKs models perform best when the bandwidth matches the data periodicity $b=\omega$. Figure \ref{fig:combined_periodic_AD}(a) shows the best configuration (with $2$ qubits, $3$ repetitions, \textit{full} entanglement), for each $\omega$ the maximum F1-score is reached on the diagonal $b=\omega$. The average F1-scores over all QK configurations (Figure \ref{fig:combined_periodic_AD}(b)) show a similar pattern for the datasets with more oscillations. For $\omega \geq 4$, the maximum F1-score is also obtained on the diagonal $b=\omega$. Low oscillations cases $\omega \leq 2$ display off-diagonal maxima. 
Restricting to the case $b=\omega$ across different dimensions $D \in \{2,4,6\}$ (Figure \ref{fig:combined_periodic_AD}(c))  the average test F1-score shows a similar pattern to the one observed for the classification task, and is stable with increasing oscillations and even increases slightly for $D\geq 4$, whereas that of the Gaussian kernels decreases with increasing $\omega$. The min-max range bars for classical and quantum methods overlap at $D=2$ and at low $\omega$ and become increasingly distinct for higher $D$ and $\omega$. 
Of note, the overall F1-score is generally low that is expected since no labels are used in training and the supports of the normal and the outlier data distributions partially overlap by construction.

For both the classification and anomaly detection tasks, these observations are consistent with the theoretical analysis of Section \ref{sec:RHKS_fourier}. As a direct consequence of Corollary \ref{corollary:IQP_RHKS_fourier} and of the definition of the bandwidth $b$, the RKHS generators contain oscillatory components whose characteristic frequencies are proportional to $b$ (see also Figure \ref{fig:RKHS_Fourier}). This oscillatory structure allows the kernel to assign relatively high similarity to points separated along the input coordinates by distances of the order of integer multiples of $1/b$, although the non-linear terms in the encoding prevent exact periodicity in general. The periodicity of the data labels is instead determined by $1/\omega$ by construction (Eq. \eqref{eq:classification_label}), so that points differing by integer multiples of $1/\omega$ along each coordinate belong to the same class. When $\omega=b$ the characteristic scales of the oscillatory structures of the kernel and the target function are therefore aligned allowing the corresponding QK methods to better capture the underlying data structure. Conversely, when $b\neq\omega$, this alignment is generally reduced, consistently with the lower F1-scores observed in Figure \ref{fig:combined_periodic_classification}(a,b) and Figure \ref{fig:combined_periodic_AD}(a,b).

The difference in the behaviour of quantum and classical kernels can also be interpreted in light of the RKHS properties discussed in Section \ref{sec:background} and Section \ref{sec:RHKS_fourier}. As $\omega$ increases, the target label function \eqref{eq:classification_label} becomes increasingly oscillatory, making the learning task more challenging. The Gaussian kernel, whose generators are localized around the training points and do not exhibit an intrinsic periodic structure, is expected to require increasingly fine spatial resolution to capture these oscillations. In contrast, the IQP kernel contains oscillatory components whose characteristic frequencies can be adjusted through the bandwidth $b$. When $b=\omega$, these components remain aligned with the periodic structure of the target function as its frequency increases. This provides a possible explanation for the observed stability of the QK performance with respect to $\omega$, in contrast to the progressive deterioration of the Gaussian kernel.

\begin{figure*}
 \centering
    \includegraphics[width=1\linewidth]{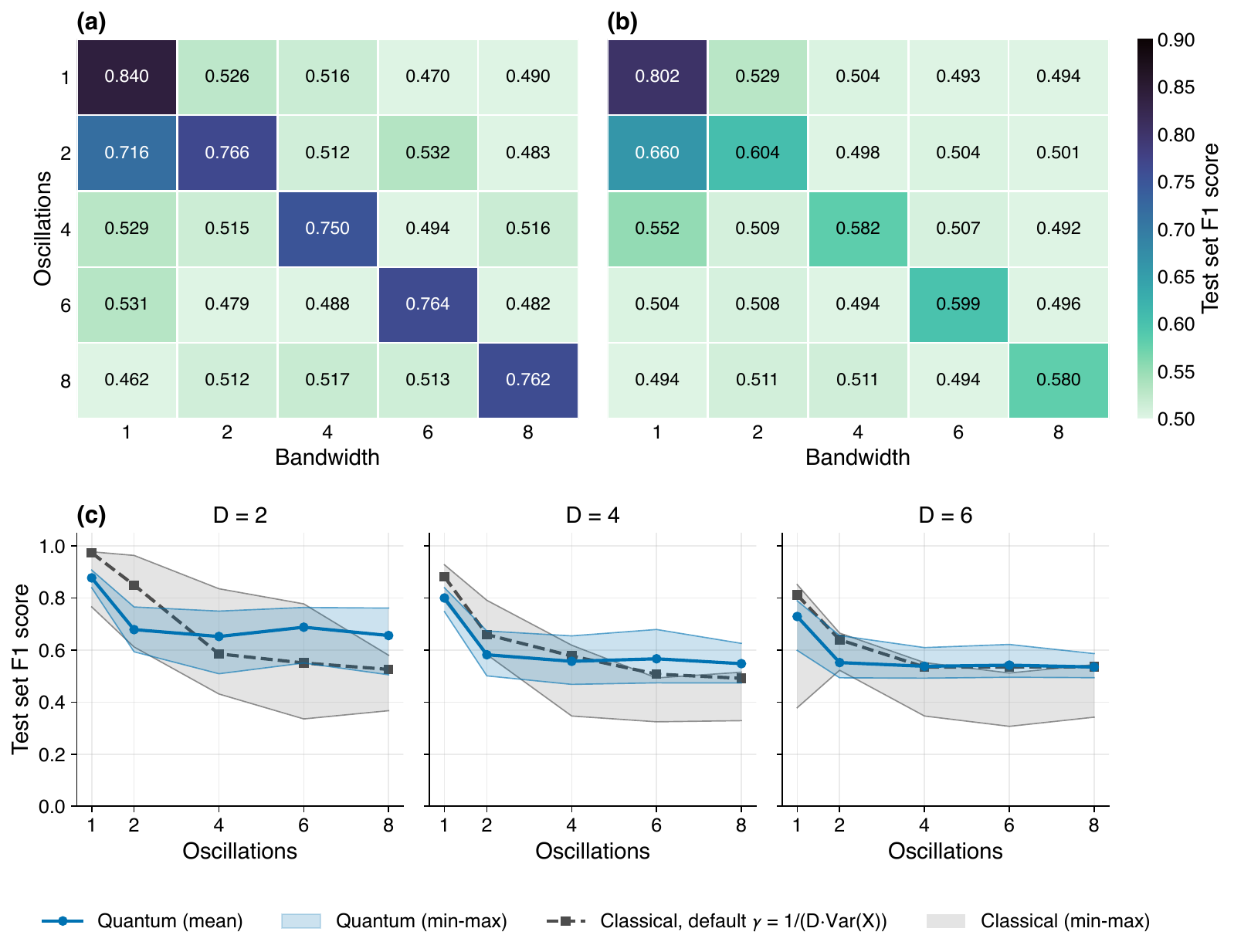}
    \caption{Test-set F1-score for the periodic classification task. (a) F1-score of a single representative QK configuration (2 qubits, 3 repetitions, \textit{linear} entanglement) as a function of the oscillation frequency $\omega$ and the bandwidth $b$. (b) F1-score averaged over all considered QK configurations (number of qubits, repetitions and entanglement strategies) on the same $(\omega,b)$ grid. In both panels the F1-score peaks along the diagonal $b=\omega$ and drops off-diagonal, a pattern more pronounced in (a) than in (b). (c) Average F1-score of QKs with $b=\omega$ (blue line) and of the classical Gaussian kernel with default $\gamma$ (dashed grey line) as a function of $\omega$, for $D \in \{2,4,6\}$. The blue area indicates the minimum and maximum F1-score of QKs across repetitions and entanglement strategies, while the grey area indicates the minimum and maximum F1-score of the classical kernel across $\gamma \in \{0.01,0.1,1,10,100\}$.}
    \label{fig:combined_periodic_classification}
\end{figure*}

 \begin{figure*}
     \centering
     \includegraphics[width=1\linewidth]{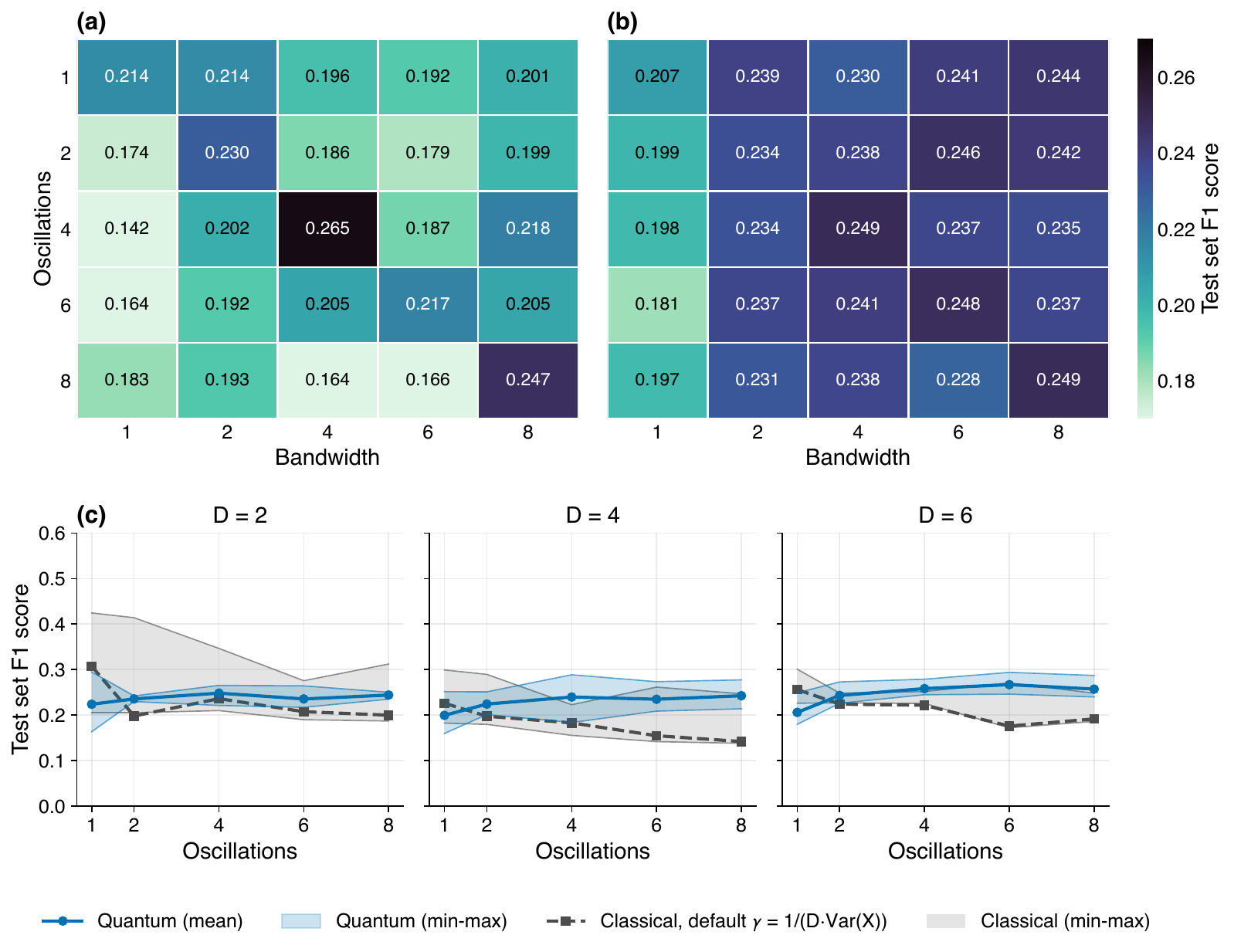}
 \caption{Test-set F1-score for the periodic anomaly detection task. (a) F1-score of a single representative QK configuration (2 qubits, 3 repetitions, \textit{full} entanglement) as a function of the oscillation frequency $\omega$ and the bandwidth $b$. (b) F1-score averaged over all considered QK configurations (number of qubits, repetitions and entanglement strategies) on the same $(\omega,b)$ grid. In (b) the lowest scores are obtained for $b=1$, while for $\omega \geq 4$ the highest score in each row is reached along the diagonal $b=\omega$. (c) Average F1-score of QKs with $b=\omega$ (blue line) and of the classical Gaussian kernel with default $\gamma$ (dashed grey line) as a function of $\omega$, for $D \in \{2,4,6\}$. The blue area indicates the minimum and maximum F1-score of QKs across repetitions and entanglement strategies, while the grey area indicates the minimum and maximum F1-score of the classical kernel across $\gamma \in \{0.01,0.1,1,10,100\}$.}
     \label{fig:combined_periodic_AD}
 \end{figure*}

 \subsection{Multimodal data} \label{sec:multimodal}
In this section, we assess the ability of QKs to capture multimodal data on the same two ML tasks tested in the previous section. To this end, we generated datasets according to a multimodal distribution in the domain $[0,1]^D$ for different values of $D \in \{2,4,6\}$. Multimodality was obtained by sampling data from a distribution which is the sum of independent Gaussian distributions centred at different points in $[0,1]^D$ and with fixed variance. Visual representations of the datasets are reported in Figure \ref{fig:datasets}(c,d). The density function of the distribution from which data were sampled is given by:
\begin{equation}
  f(x) = \frac{1}{C} \sum_{i=1}^C  f_{c_i,\sigma}(x)  
  \label{eq:density_fucntion}
\end{equation} 
where $f_{c_i,\sigma}$ is the density of a Gaussian function centred at $c_i$ with standard deviation $\sigma$.
As a Gaussian mixture, this dataset does not contain high-frequency structure, making it representative of the small- and intermediate-bandwidth regimes discussed in Section \ref{sec:RHKS_fourier}. The degree of multimodality was controlled by the number of centres $C \in \{2,4,8\}$, and the ground truth labels were assigned to each of the datasets depending on the ML task considered. For the binary classification task, two sets of data were generated by \eqref{eq:density_fucntion} with different, randomly selected, centres. Each set of data was assigned to a binary class. In this way, two multimodal classes were obtained, each with $C$ modes and a balanced dataset. For anomaly detection, data were generated by using \eqref{eq:density_fucntion} as normal data and adding $20\%$ of data uniformly generated in $[0,1]^D$ as anomalies. 

In this experiment we explored QK configurations with bandwidth values $b \in \{1/4, 1/2, 1, 3/2, 2\}$. We used the same number of repetitions and entanglement types tested for the periodic data scenario (i.e., $r \in \{1,2,3\}$, and entanglement strategies \textit{linear} and \textit{full}) and computed the classical RBF kernels with the same settings employed for the periodic dataset. We then fitted and evaluated kernel-based ML models using respectively SVMs and OCSVMs for binary classification and anomaly detection. 

Figure \ref{fig:multimodal_results} shows the test-set F1-score obtained by QKs as a function of the bandwidth for different dimensions and degrees of multimodality. The figure also reports the F1-score for  classical kernel methods for a comparison. 

In classification (Figure \ref{fig:multimodal_results}(a)) QKs achieve near perfect performance (F1-score close to 1) for $D \geq 4$ or for $C=2$, for intermediate values of the bandwidth $b = 1$. In the same setting, the performance deteriorates for larger and smaller values of the bandwidth. In the case of $D=2$ and $C \geq 4$ the F1-scores are generally lower but still peak at intermediate bandwidth values $b \in \{1,1.5\}$. In all cases, the plot describing the dependency of the F1-scores on the bandwidth describes an arched curve peaking at intermediate bandwidth values. Classical kernels with the default $\gamma$ perform similarly to the best quantum kernels and the quantum and classical performance min-max ranges largely overlap. 

In anomaly detection (Figure \ref{fig:multimodal_results}(b)), the effect of the bandwidth is more pronounced, with the mean  F1-score showing a clear peak at intermediate bandwidths for all degrees of multimodality $C$. The maximum mean F1-score is reached at $b=1$  in all settings ($0.64$-$0.80$), except for $D=2$ and $C>2$ where it shifts to $b=1.5$ ($0.43$-$0.49$). Additionally, a higher degree of multimodality makes anomalies harder to detect, as for a given $D$ the maximum mean F1-score decreases with $C$, from $0.64$ to $0.43$ for $D=2$, from $0.79$ to $0.65$ for $D=4$ and from $0.80$ to $0.69$ for D=6.
Gaussian kernels with default $\gamma$ reach a F1-score comparable to the one of QKs with intermediate $b$ values. Additionally, quantum and classical performance ranges largely overlap.

These results are consistent with the qualitative behaviour of the RKHS generators discussed in Section \ref{sec:RHKS_fourier}. At small bandwidth values, the IQP kernel exhibits a localized structure similar to that of the Gaussian kernel, primarily capturing relationships between nearby points. At intermediate bandwidths $b \in \{1,1.5\}$, the RKHS generators develop a multimodal structure, allowing the kernel to assign relatively high similarity also to points separated in the original input space. This property can be beneficial for learning multimodal distributions, where points belonging to the same class or to the normal data distribution may be concentrated in distinct, spatially separated regions. At larger bandwidths, the increasingly oscillatory structure of the kernel introduces additional variations in similarity that are not aligned with the underlying data distribution. This provides an explanation for the observed peak in performance at intermediate bandwidths. 
The stronger dependence on the bandwidth observed in anomaly detection may be explained by two factors. First, anomaly detection with OCSVM does not use class labels in training. It therefore relies on the similarity structure induced by the kernel to learn the distribution of normal data. Since normal data were chosen by construction to be multimodal (see Eq. \eqref{eq:density_fucntion}), intermediate values of $b$, which give rise to multimodal RKHS generators, may be particularly suitable for capturing the structure of the normal-data distribution. 
Second, the classification task considered here appears to be relatively easy, with high F1-scores already achieved at small bandwidth values. This suggests that the local similarity structure available in this regime is already sufficient to separate the two classes, leaving limited room for improvement at intermediate bandwidths. 

\begin{figure}[!h]
    \centering
    \includegraphics[width=1\textwidth]{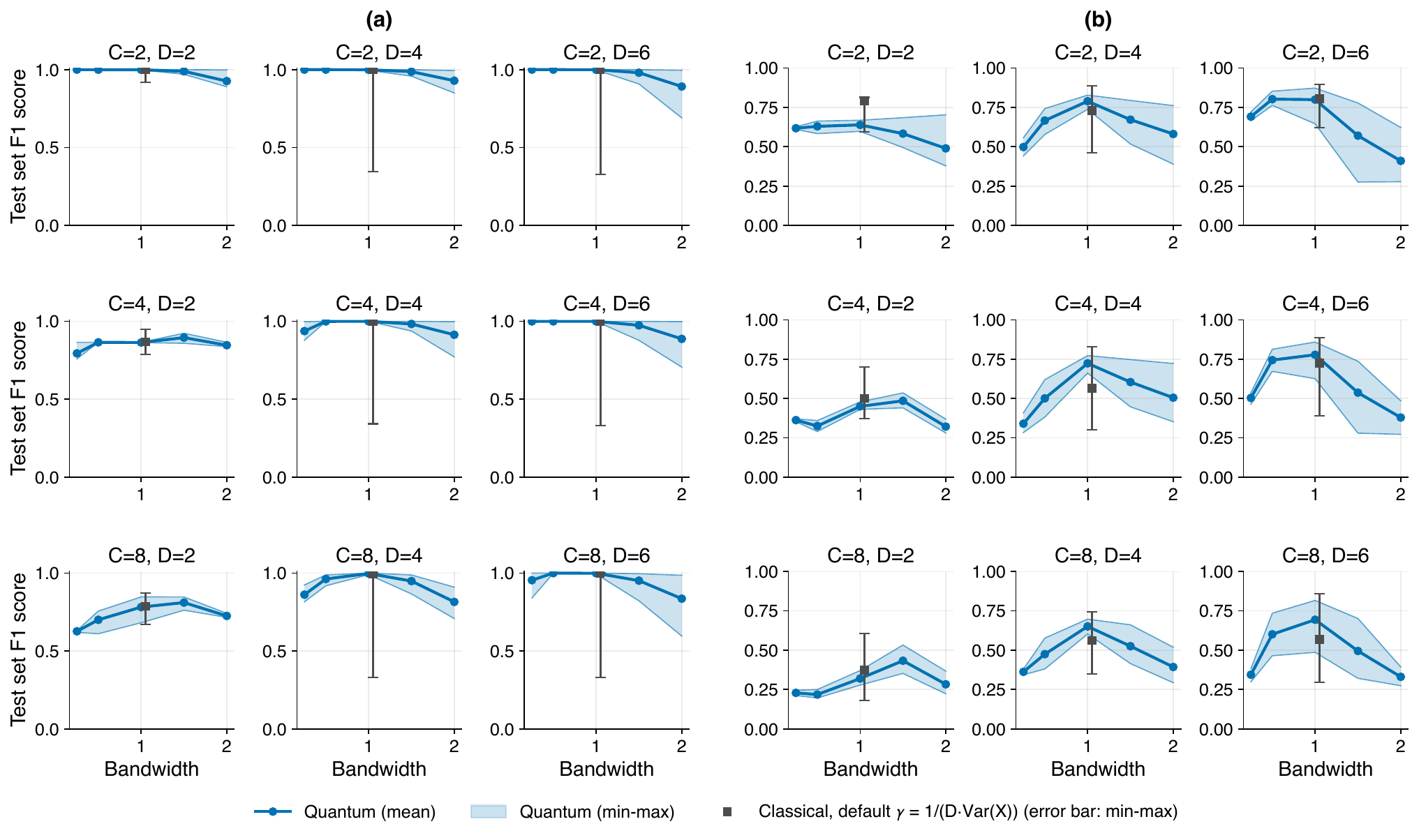}
    \caption{F1-scores of classical and QK methods on datasets with varying degree of multimodality, for the supervised classification task (a) and the anomaly detection task (b). Each plot corresponds to a dataset generated with a different number of centres per class $C$ (degree of multimodality) and dimensionality $D$ (equal to the number of qubits); the horizontal axis displays the bandwidth value $b$. Since the classical RBF kernel does not depend on $b$, its performance is shown as a single grey point corresponding to the default $\gamma$ setting (plotted at the mean bandwidth for reference), with a bar spanning the min-max performance across $\gamma \in \{0.01,0.1,1,10,100\}$.}
    \label{fig:multimodal_results}
\end{figure}

We also observe that the other two hyperparameters, which are the number of repetitions $r$ and the entanglement strategy, have a non-negligible  but less straightforward impact on QK performance. The performance obtained for each value of $r$ and each entanglement strategy is reported separately in Supplementary Figures S1--S4. As discussed in Section \ref{sec:RHKS_fourier}, $r$ determines the accessible multiples of the base frequency and the entanglement influences the relative weights of the terms in the RKHS Fourier expansion. Neither parameter directly controls the base frequency, which is determined by the bandwidth $b$. 
The datasets considered in these experiments were specifically designed to investigate periodicity and multimodality, two structural characteristics of the RKHS generators that are primarily controlled by $b$ (see also Figures \ref{fig:RKHS_quantum} and \ref{fig:RKHS_Fourier}). The effects of $r$ and entanglement are therefore less directly related to the properties of the datasets and more difficult to characterise systematically. Nevertheless, the variability across QK configurations shown in Figures \ref{fig:combined_periodic_classification}(c), \ref{fig:combined_periodic_AD}(c), and \ref{fig:multimodal_results} indicate that their influence on performance is non-negligible, even for the optimal bandwidth selection. This is particularly evident for the multimodal anomaly detection task (Figure \ref{fig:multimodal_results}(b)) and for periodic classification with $D=2$ (Figure \ref{fig:combined_periodic_classification}(c)). These observations suggest that, although the bandwidth can be selected based on the structural characteristics of the data, the remaining hyperparameters should also be optimised using standard hyperparameter tuning procedures, as commonly done in classical machine learning methods.
\section{Conclusions} 
In this article, we presented a theoretical characterisation of the expressivity of quantum kernel methods based on the IQP feature map. 
By extending  the framework of \cite{Schuld_data_encoding}, we proved that the generators of the RKHS associated with IQP kernels can be expressed as a composition of periodic trigonometric functions acting on polynomially transformed input variables, with a finite set of accessible frequencies determined by the feature-map construction and scaled by the bandwidth hyperparameter.
This representation implies that the qualitative behaviour of the kernel shifts with the bandwidth, progressing from Gaussian-like at small values, to multimodal at intermediate values, to periodic at large values. We validated these findings through numerical experiments on two families of synthetic datasets, each designed to probe these different regimes: one with periodic structure and one with multimodal structure, for supervised classification or unsupervised anomaly detection tasks. We further extended our experiments across different data dimensions and IQP feature-map configurations (i.e. entanglement level and repetitions).
Three distinct operating regimes emerged: for small bandwidth values, IQP kernels behave qualitatively similarly to Gaussian kernels, as previously shown by \cite{Alban_2025}; for intermediate bandwidth values, the RKHS generators become multimodal, enabling the kernel to capture longer-range interactions between data points and yielding improved performance on datasets with multimodal structure. Finally, for sufficiently large bandwidth values, the periodic nature of the RKHS generators becomes dominant, allowing quantum kernel methods to effectively model periodic patterns whenever the periods present in the data are aligned with those induced by the feature map.
While the theoretical findings highlight specific links between the properties of the data and those of IQP kernels, we acknowledge that the numerical experiments are restricted to synthetic datasets with controlled periodic or multimodal structure; assessing how these regimes translate to real-world data is a fundamental next step that we plan to investigate in future work. Furthermore, the theoretical characterisation applies specifically to IQP feature maps, while further investigation is required to extend it to other families of quantum feature maps.
Nevertheless, the search of a connection between the structural properties of data and the mathematical properties of quantum models is a central open question for the field \cite{bowles2024betterclassicalsubtleart}. In this regard, the here presented theoretical and experimental results establish a direct connection between the structure of the RKHS, the hyperparameters of the IQP feature map, and the characteristics of the data. Rather than viewing the bandwidth solely as a tuning parameter, we have found that it determines the RKHS structure and the classes of datasets benefit from IQP-based quantum kernel methods.

\section*{Acknowledgement}
We acknowledge the CINECA award under the ISCRA initiative, for the availability of high performance computing resources and support.
\newline
\section*{Declarations}
\subsection*{Funding}
The European Union - Next Generation EU, in the context of PNRR, Investment 1.5 Ecosystems of Innovation, Project Tuscany Health Ecosystem (THE), ECS00000017, Spoke 3. CUP: B83C22003930001 (to VR and RDA) 
\subsection*{Conflict of interest}
The authors have no competing interests to declare that are relevant to the content of this article.
\subsection*{Data and code availability }
The data, code, and experimental results associated with this project are available on GitHub at the following link: \url{https://github.com/stefanomelchionna/Expressivity_of_IQP_Feature_Map_QKs_via_RKHS/}
\subsection*{Authors contribution}
SM and VR conceived and conducted the study, performed computation and analyses, interpreted the results, and drafted the manuscript; RDA supervised the study, interpreted the results, and reviewed the manuscript. All authors approved the final manuscript.
\bibliography{bibliography}
\clearpage
\onecolumngrid
\section*{Supplementary Material}
\label{sec:supp_hyperparameters}

\renewcommand{\thefigure}{S\arabic{figure}}
\setcounter{figure}{0}

This section reports the test-set F1-score of QKs separately for each number of repetitions $r \in \{1,2,3\}$ and each entanglement strategy (\textit{linear} and \textit{full}), for the periodic datasets with $b=\omega$ (Figures \ref{fig:supp_S1} for classification and \ref{fig:supp_S2}  for anomaly detection) and for the multimodal datasets (Figures \ref{fig:supp_S3}  for classification and \ref{fig:supp_S4} for anomaly detection).

\begin{figure}[!h]
    \centering
    \includegraphics[width=0.75\textwidth]{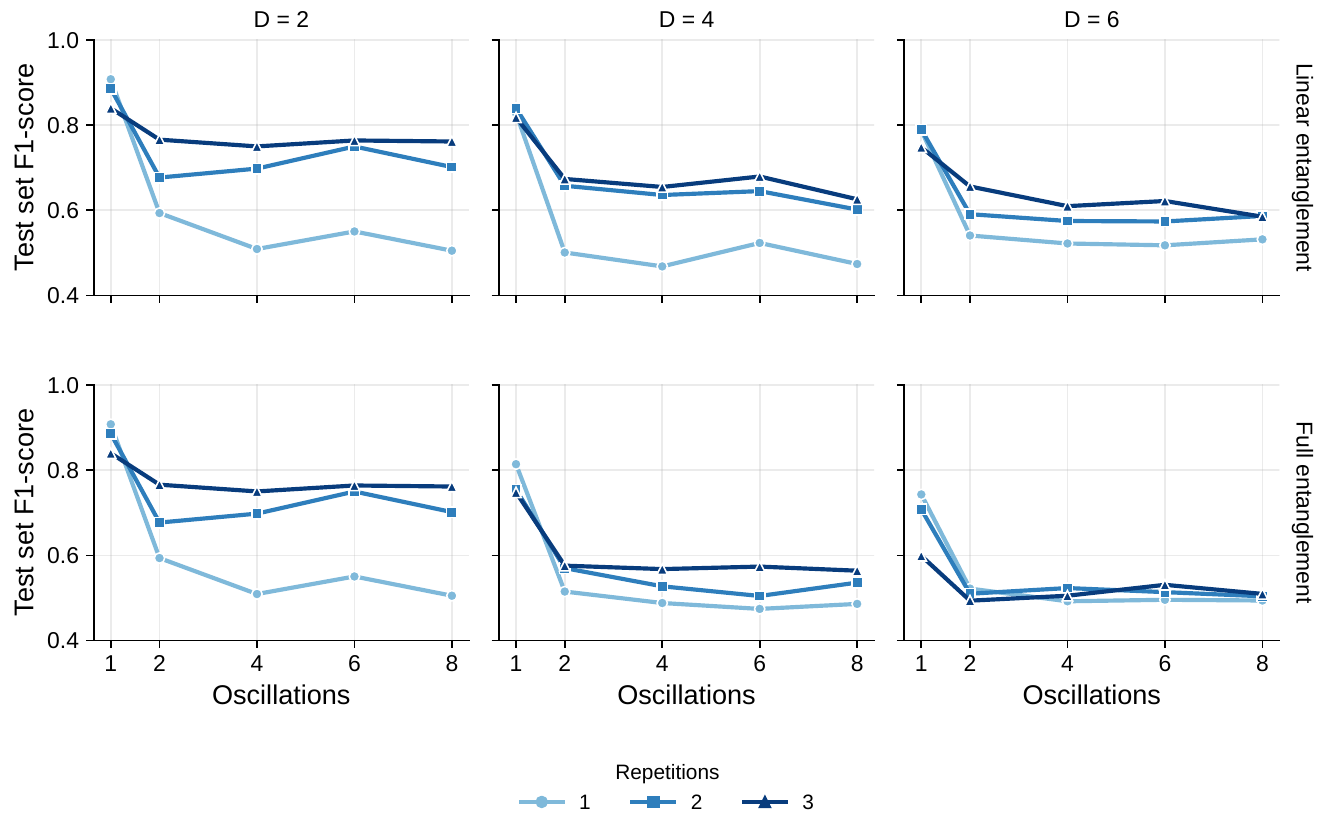}
    \caption{Test-set F1-score for the periodic classification task with $b=\omega$, as a function of the oscillation frequency $\omega$. Each column corresponds to a different dimensionality $D$ (equal to the number of qubits), the top row to \textit{linear} and the bottom row to \textit{full} entanglement. Lines are coloured according to the number of repetitions $r \in \{1,2,3\}$.}
    \label{fig:supp_S1}
\end{figure}

\begin{figure}[!h]
    \centering
    \includegraphics[width=0.75\textwidth]{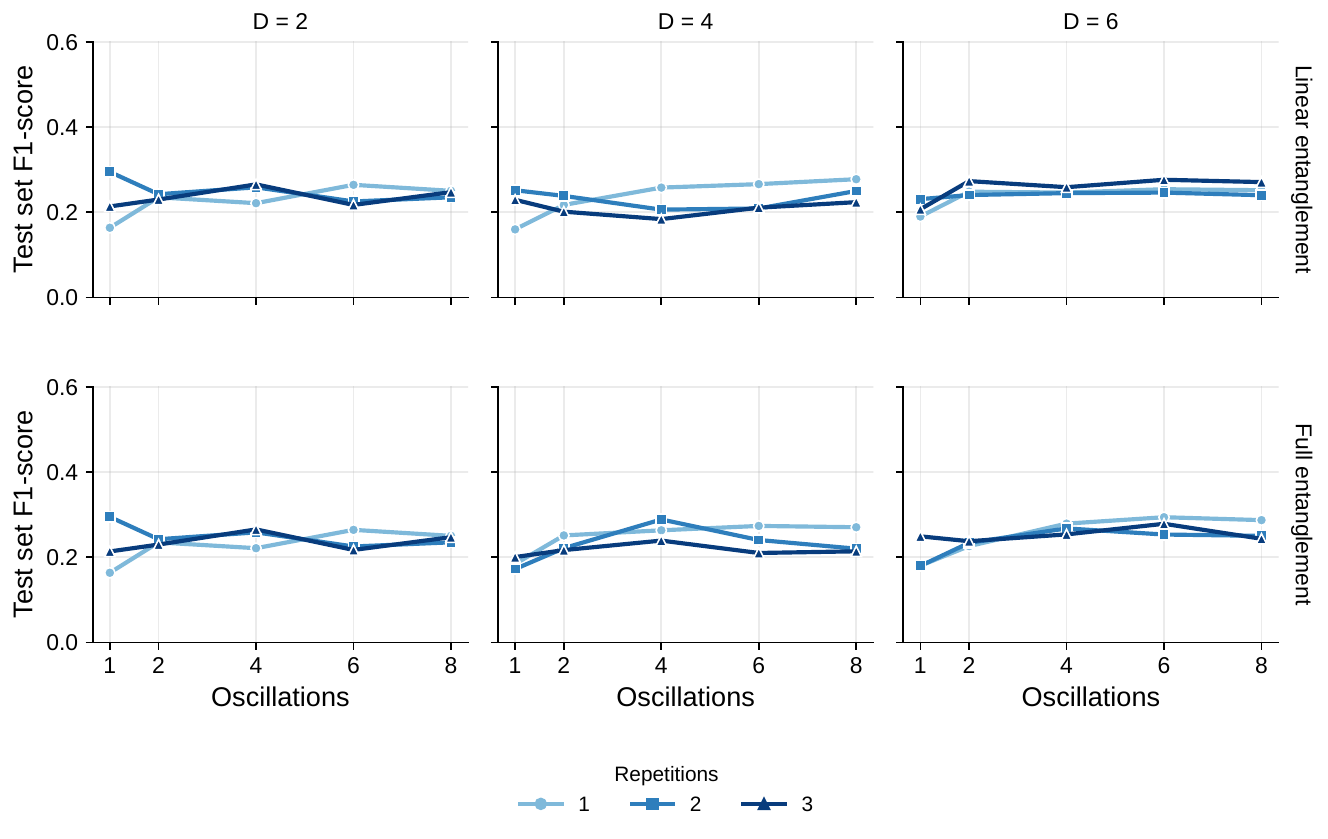}
    \caption{Test-set F1-score for the periodic anomaly detection task with $b=\omega$, as a function of the oscillation frequency $\omega$. Each column corresponds to a different dimensionality $D$ (equal to the number of qubits), the top row to \textit{linear} and the bottom row to \textit{full} entanglement. Lines are coloured according to the number of repetitions $r \in \{1,2,3\}$.}
    \label{fig:supp_S2}
\end{figure}

\begin{figure}[!h]
    \centering
    \includegraphics[width=0.8\textwidth]{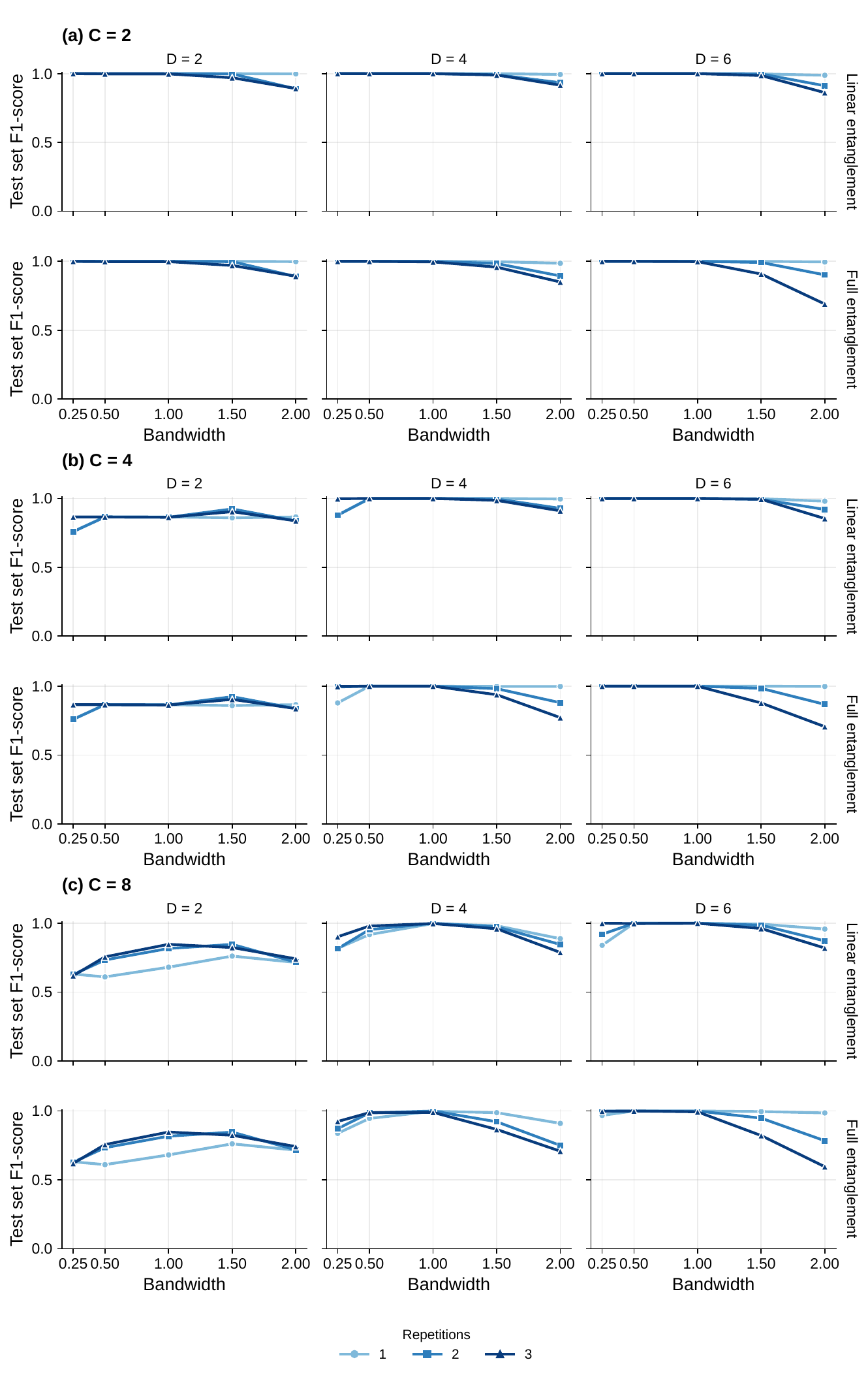}
    \caption{Test-set F1-score for the multimodal classification task as a function of the bandwidth $b$. Panel (a), (b) and (c) correspond to datasets generated with $C=2$, $4$ and $8$ centres per class (degree of multimodality). In each panel, columns correspond to a different dimensionality $D$ (equal to the number of qubits), the upper row to \textit{linear} and the lower row to \textit{full} entanglement. Lines are coloured according to the number of repetitions $r \in \{1,2,3\}$.}
    \label{fig:supp_S3}
\end{figure}

\begin{figure}[!h]
    \centering
    \includegraphics[width=0.8\textwidth]{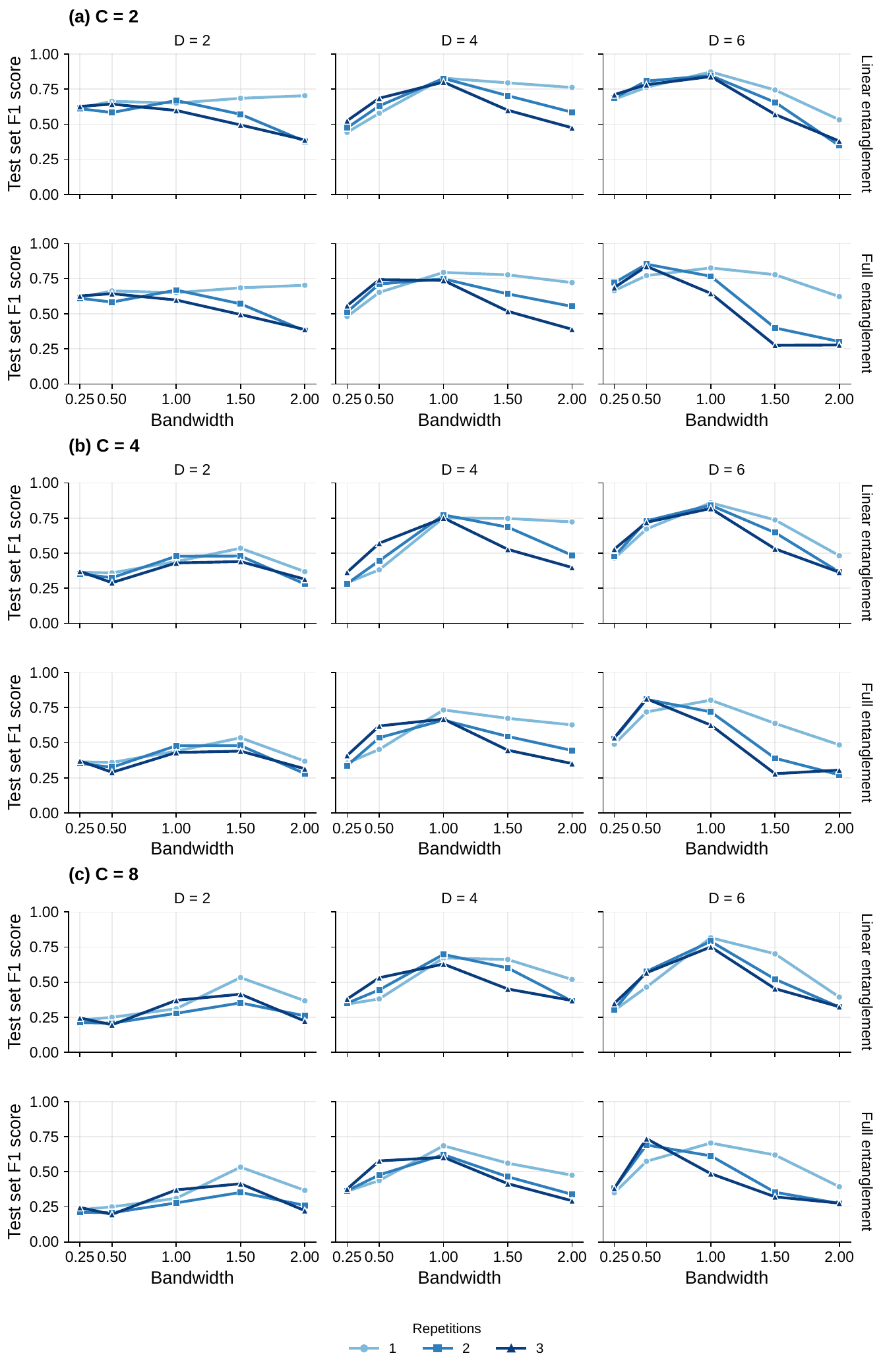}
    \caption{Test-set F1-score for the multimodal anomaly detection task as a function of the bandwidth $b$. Panel (a), (b) and (c) correspond to datasets generated with $C=2$, $4$ and $8$ centres (degree of multimodality). In each panel, columns correspond to a different dimensionality $D$ (equal to the number of qubits), the upper row to \textit{linear} and the lower row to \textit{full} entanglement. Lines are coloured according to the number of repetitions $r \in \{1,2,3\}$.}
    \label{fig:supp_S4}
\end{figure}
\end{document}